\documentclass[a4paper,11pt,]{article}
\usepackage[english]{babel}
\usepackage[latin1]{inputenc}
\usepackage{amsmath}                    
\usepackage{amsthm}                     
\usepackage{amsfonts}                   
\usepackage{bbm}                        
\usepackage{natbib}                     
\usepackage{mathtools}                  
\mathtoolsset{showonlyrefs}             
\usepackage[title, titletoc]{appendix}  
\usepackage{booktabs}
\usepackage{tikz}
\usepackage{blkarray}                   

\usepackage{multirow}
\usepackage{enumerate}                  
\usepackage{float}
\usepackage[hidelinks]{hyperref}        
\usepackage{algorithm}

\newtheorem{thm}{Property}

\newtheorem{lem}{Lemma}
\newtheorem*{prop}{Proposition}
\newtheorem{rem}{Remark}

\def\Xb{{\bf X}}

\def\bb{{\bf b}}

\def\xb{{\bf x}}
\def\yb{{\bf y}}

\def\Pr{{\rm Pr}}
\def\E{{\rm E}}
\def\Var{{\rm Var}}
\def\Bias{{\rm Bias}}

\def\betab{\boldsymbol{\beta}} 
\def\lambdab{\boldsymbol{\lambda}} 
\def\phib{\boldsymbol{\phi}} 
\def\psib{\boldsymbol{\psi}} 

\def\Sigmab{\boldsymbol{\Sigma}} 

\def\1{\mathbbm 1}

\def\bknn{ [bk] }
\def\knn{ [k] }

\def\O#1#2{O\left( \frac{#1}{#2} \right)}

\reversemarginpar 

\title{Balanced $k$-Nearest Neighbor Imputation}%
\author{Caren Hasler\thanks{\emph{Address for correspondence}: Caren Hasler, Institute of Statistics, University of Neuch\^atel, 2000 Neuch\^atel, Switzerland.\newline E-mail: \texttt{caren.hasler@unine.ch}}~ and Yves Till\'e\\
\emph{University of Neuch\^atel, Switzerland}}

\begin{document}

\maketitle


\begin{abstract}
In order to overcome the problem of item nonresponse, random imputation methods are often used because they tend to preserve the distribution of the imputed variable. Among the random imputation methods, the random hot-deck has the interesting property of imputing observed values. A new random hot-deck imputation method is proposed. The key innovation of this method is that the selection of donors is viewed as a sampling problem and uses calibration and balanced sampling. This approach makes it possible to select donors such that if the auxiliary variables were imputed, their estimated totals would not change. As a consequence, very accurate and stable totals estimations can be obtained. Moreover, the method is based on a nonparametric procedure. Donors are selected in neighborhoods of recipients. In this way, the missing value of a recipient is replaced with an observed value of a similar unit. This new approach is very flexible and can greatly improve the quality of estimations. Also, this method  is unbiased under very different models and is thus resistant to model misspecification. Finally, the new method makes it possible to introduce edit rules while imputing.
\end{abstract}

\section{Introduction}

Nonresponse is an important problem in surveys. Indeed, the error caused by nonresponse on estimates can be more severe than the error caused by the sampling design. Nonresponse arises when a sampled unit does not respond to one or more items of a survey. One differentiates item nonresponse (a sampled unit does not respond to a particular question) from unit nonresponse (a sampled unit does not respond to the entire survey). Reweighting procedures are often used to deal with unit nonresponse whereas imputation methods are used to treat item nonresponse. Imputation denotes a procedure to replace a missing value with a substituted one.

Imputation methods are classified as either deterministic or random. Deterministic refers to imputation methods that yield the same imputed values if the imputation is repeated. Deterministic imputation methods include ratio imputation, regression imputation, respondent mean imputation, and nearest neighbor imputation. Deterministic imputation methods produce good totals estimations. Nevertheless they often fail to estimate quantiles. Random imputation refers to  methods that yield different imputed values if the imputation is repeated. Random imputation methods include among others multiple imputation presented in \cite{rub:87}, imputation with added residuals considered in \cite{cha:dev:haz:11}, and random $k$-nearest neighbor imputation ($k$NNI) \cite[see among others][]{dah:07}. Unlike deterministic imputation methods, random imputation methods offer the advantage of tending to preserve the distribution of the imputed variable. Nevertheless such methods imply the presence of an additional amount of variance due to the randomness of imputation, which is called \emph{imputation variance}. Many authors have been interested in minimizing imputation variance.  For instance, \cite{kal:kis:81,kal:kis:84} proposed two ways to reduce imputation variance. The first way consists of selecting donors among the respondents without replacement rather than with replacement. The second way consists of first constructing strata using the respondents' values of the variable of interest and selecting proportionate stratified samples to act as donors. To the same end, \cite{che:rao:sit:00} proposed adjustment of the imputed values; \cite{kim:ful:04} and \cite{ful:kim:05} used fractional hot-deck imputation, and \cite{cha:dev:haz:11} introduced a class of balanced random imputation methods
that consists of randomly selecting residuals while satisfying given constraints.

Imputation methods can alternatively be classified as either donor or predicted value. Donor imputation methods replace the missing value of a nonrespondent with the observed value of a respondent. The unit providing the value is called a \emph{donor} and the unit receiving the value is called a \emph{recipient}. A hot-deck method is a donor imputation method where a missing value is replaced with an observed value extracted from the same survey. Such a method is particularly of interest because it imputes feasible and observed values. The reader can, for instance, refer to~\cite{and:lit:10} for a review of hot-deck imputation. In contrast, predicted value imputation methods use function of the respondents values to predict the missing values.


In this paper, a new method of random hot-deck imputation is proposed: the balanced $k$-nearest neighbor imputation method (b$k$NNI). The main feature of this method is that the selection of donors is viewed as a sampling problem and uses calibration and balanced sampling. This makes it possible to select donors such that if the auxiliary variables were imputed, their estimated totals would not change. Moreover, this method is based on a nonparametric procedure. Indeed, it provides donors selected in neighborhoods of recipients. In this way, the gap due to one unit's missing value is filled with a similar unit's observed value.
The novelty of the proposed method lies not only in the fact that the selection of donors uses balanced sampling but also in the fact that this is paired with a nonparametric selection of donors. Considered together in the same procedure, these two features imply a robustness in terms of model misspecification. Moreover, this method uses a methodology that makes it possible to take  edit rules into account while imputing. The proposed method is also particularly effective, produces negligible imputation variance and a quasi-null bias in specified cases.

This paper is organized as follows. In Section~\ref{section:notation}, notation and concepts of nonresponse are reviewed. A methodology for random hot-deck imputation methods is introduced in Section~\ref{section:methodology:hotdeck}. Section~\ref{section:balanced:knn} focusses on the presentation of the b$k$NNI. Section~\ref{section:estimation:imputation:variance} introduces a formula to approximate the conditional imputation variance of the total when the new method is applied. Section~\ref{section:model} describes the models underlain by the method and studies the asymptotic properties of the total estimator. Then, in Section~\ref{section:simulation}, the performance of the new imputation method and the accuracy of the proposed estimator for imputation variance are tested through a simulation study. A short discussion concludes the paper in Section~\ref{section:conclusion}.

\section{Notation and concepts of nonresponse}\label{section:notation}

Consider a finite population $U = \left\{1,2,\ldots,i,\ldots,N\right\}$ and suppose that the target is the variable of interest $\yb = \left( y_1, y_2, \ldots, y_i, \ldots, y_N \right)^\top$. In a first phase, a random sample $S$ of size $n$ is drawn from $U$ with a given sampling design $p\left(\cdot\right)$ where $p(s)= \Pr \left( S = s\right)$ for $s \subset U$. Let $\pi_i = \Pr \left( i \in S \right)$ denote the first order inclusion probability of unit $i$ and let $d_i = 1/ \pi_i$ denote its Horvitz-Thompson weight \citep{hor:tho:52}. If a census is considered, the inclusion probabilities and the design weights are equal to 1. A vector $\xb_i = \left( x_{i1}, x_{i2}, \ldots, x_{iQ} \right)^\top$ of $Q$ auxiliary variables is assumed to be known for each unit $i$ in the sample $S$. In what follows, it is supposed that one of the auxiliary variables is constant. In a second phase, a subset of respondents $S_r \subset S$ is obtained from $S$ with a usually unknown conditional distribution $q\left( S_r | S\right)$. The values $y_i$ of the variable of interest are known for the units of $S_r$ only. Let $S_m = S \setminus S_r$ denote the complement of $S_r$ in $S$, i.e. the subsample of $S$ containing the units with missing data (the nonrespondents). The respective sizes of these subsets are $n_r$ and $n_m$ with $n_r + n_m = n$. For $i \in S$, let $r_i$ be the response indicator variable
\begin{align*}
    r_i = \left\{
            \begin{array}{ll}
              1 & \hbox{if unit $i$ belongs to $S_r$,} \\
              0 & \hbox{otherwise.}
            \end{array}
          \right.
\end{align*}
It is supposed that the units respond independently from each other. For each unit $i \in S$, $r_i$ is therefore generated from a Bernoulli random variable with parameter $\theta_i = \Pr \left( i \in S_r | i \in S \right)$. The parameter $\theta_i$ represents the response propensity of unit $i$ and is usually unknown. Hence, the conditional distribution $q\left( S_r | S\right)$ is a Poisson sampling design, i.e.
\begin{align*}
    q\left( S_r | S\right) = \prod_{i \in S_r} \theta_i  \prod_{i \in S_m} \left( 1 - \theta_i \right).
\end{align*}
Three types of nonresponse mechanisms exist: uniform, ignorable, and non-ignorable. A uniform nonresponse mechanism is a nonresponse mechanism where each unit of the population has the same response propensity, i.e. $\theta_i = \theta$ for each $i \in U$. It is in this case said that the data is missing completely at random (MCAR). An ignorable nonresponse mechanism~\citep{rub:76} is a nonresponse mechanism where the response propensity $\theta_i$ does not depend on the variable of interest once the auxiliary variables have been taken into account. In the case of an ignorable nonresponse mechanism, the data is said to be missing at random (MAR). Finally, a non-ignorable nonresponse mechanism is a nonresponse mechanism where the response propensity $\theta_i$ depends on the variable of interest. The data is in this case said to be not missing at random (NMAR). In a third phase, nonresponse can be corrected through imputation.
Imputed values $y_j^*$, $j \in S_m$ are drawn with a conditional distribution
\begin{align*}
    I \left( y_j^* | S, S_r \right).
\end{align*}

The aim is to estimate the population total
\begin{align*}
Y = \sum_{i \in U} y_i,
\end{align*}
of the variable of interest $\yb$. In the case of complete response, the estimator
\begin{align*}
\widehat{Y} = \sum_{i \in S} d_i y_i, 
\end{align*}
is a design-unbiased estimator of $Y$. In the presence of nonresponse, the previous estimator is intractable and the imputed estimator
\begin{align*}
\widehat{Y}_I = \sum_{i \in S_r} d_i y_i + \sum_{j \in S_m} d_j y_j^*, 
\end{align*}
is used. Moreover, consider
\begin{align*}
    \Xb             &= \sum_{i \in U} \xb_i, \\
    \widehat{\Xb}   &= \sum_{i \in S} d_i \xb_i, \\
    \widehat{\Xb}_I &= \sum_{i \in S_r} d_i \xb_i + \sum_{j \in S_m} d_j \xb_j^*,
\end{align*}
where $\xb_j^*$ represents the imputed value we would have obtained for $j \in S_m$ if we were to impute the auxiliary variables. For instance, suppose hot-deck imputation is used. For each nonrespondent $j \in S_m$, a donor $i \in S_r$ is chosen. The imputed value $y_j^*$ for unit $j \in S_m$ is therefore the donor's observed value $y_i$. In this case, the imputed value $\xb_j^*$ is the observed value $\xb_i$ of the same donor as the one used to obtain $y_j^*$.

In the presented framework, the bias and variance of an imputed estimator $\widehat{\theta}_I$ (for a total or another statistic $\theta$) are given by
\begin{align}
    \Bias \left( \widehat{\theta}_I \right)     &= \E_p \E_q \E_I \left( \widehat{\theta}_I - \theta \right),\nonumber\\
    \Var \left( \widehat{\theta}_I \right)      &= \Var_p \E_q \E_I \left( \widehat{\theta}_I \right)
                                                    + \E_p \Var_q \E_I \left( \widehat{\theta}_I \right) + \E_p \E_q \Var_I \left( \widehat{\theta}_I \right),
                                                    \label{total:variance}
\end{align}
where the subscripts $p$, $q$ and $I$ indicate respectively the expectations and variances with regards to the sampling mechanism, with regards to the nonresponse mechanism, and with regards to the imputation mechanism. The first term in Expression~\eqref{total:variance} represents the sampling variance, the second term represents the nonresponse variance and the last term represents the imputation variance.

\section{Methodology for random hot-deck imputation methods}\label{section:methodology:hotdeck}

In this section, we propose an original formalization for random hot-deck donor imputation. The proposed method is presented by means of this formalization, but this last one can be used for any random hot-deck method. Random hot-deck imputation consists of replacing a missing value with an observed value extracted from the same survey. For each nonrespondent, a donor is hence randomly chosen among the respondents. Consequently, a random hot-deck imputation can be achieved through the realization of a random matrix $\phib = \left(\phi_{ij}\right)$, $(i,j) \in S_r \times S_m$ such that
\begin{align*}
\phi_{ij} &=
\left\{
\begin{array}{ll}
1 & \hbox{if nonrespondent $j$ is imputed by respondent $i$,} \nonumber\\
0 & \hbox{otherwise,}
\end{array}
\right.
\end{align*}
which can be rewritten
\begin{equation}\label{matrix:imputation:definition}
    \phi_{ij} = \mathbbm{1}_{y_j^* = y_i}.
\end{equation}
As exactly one donor must be selected for each nonrespondent, $\phib$ must satisfy
\begin{equation}\label{equation:constraint:sum:imputation}
    \sum_{i \in S_r} \phi_{ij} = 1, \quad \mbox{for each }j \in S_m.
\end{equation}
It is here considered that a respondent can be used to impute several nonrespondents. Therefore, no conditions are imposed on
$$
    \sum_{j \in S_m} \phi_{ij},
$$
for $i \in S_r$. Taking the conditional expectation of Equation~\eqref{matrix:imputation:definition} generates a matrix of imputation probabilities $\psib = \left(\psi_{ij}\right)$, $(i,j) \in S_r \times S_m$
\begin{equation}\label{equation:link:psi:phi}
\psi_{ij} = \E_I\left(\phi_{ij}\right) = \E_I\left( \mathbbm{1}_{y_j^* = y_i}\right) =  \Pr \left( y_j^* = y_i | S, S_r\right).
\end{equation}
By definition, $\psib$ satisfies
\begin{align}
    &\sum_{i \in S_r} \psi_{ij} = 1 \quad \mbox{for each } j \in S_m,   \label{equation:constraint:P:sum1}\\
    & \psi_{ij} \geq 0  \quad \mbox{for each } (i,j) \in S_r \times S_m. \label{equation:constraint:P:coeff01}
\end{align}

The considered methodology for random hot-deck donor imputation is therefore operated in two stages. In the first stage, the matrix of imputation probabilities $\psib$ is defined. Then, in the second stage, a realization of the matrix of imputation $\phib$ is carried out. This methodology has, among other things, the interesting property to make it possible to implement edit rules in the imputation process to correct the data at the record level. Indeed, suppose that the value of the variable of interest $y_i$ of a respondent $i \in S_r$ is, for some reason, inconsistent with a nonrespondent $j \in S_m$. To take this inconsistency into account, it is sufficient to set a zero in coefficient $\psi_{ij}$ of the matrix of imputation probabilities $\psib$. In this way, indeed, unit $i \in S_r$ is removed from the set of possible donors for nonrespondent $j \in S_m$. Note that it might then be required to rescale the column of matrix $\psib$ corresponding to nonrespondent $j$ in order to ensure that this one still sums up to 1. Two examples of application of this methodology are provided below.

\subsection*{Example 1: simple random imputation method (with replacement)}\label{section:srswr}

Simple random imputation consists of selecting, for each nonrespondent, a donor among the respondents. The donors are selected randomly, with replacement (a donor can be used several times), and with equal probabilities. This results in the matrix of imputation probabilities $\psib^{[srs]} = \left( \psi_{ij}^{[srs]} \right)$, $(i,j) \in S_r \times S_m$ with
\begin{equation}\label{equation:psi:srs}
\psi_{ij}^{srs} = \frac{1}{n_r}.
\end{equation}
Notation $[srs]$ in superscript of the matrix $\psib$ means that the matrix linked to simple random imputation is considered.

\subsection*{Example 2: random $k$-nearest neighbor imputation method}\label{section:random:knn}

The \emph{$k$-nearest neighbors} of a nonrespondent unit $j \in S_m$ are here defined as its $k$ most similar respondents units $i \in S_r$, i.e.
\begin{align*}
    \mbox{knn}(j) = \left\{ i \in S_r | \mbox{rank} \left( d(i,j) \right) \leq k \right\},
\end{align*}
where $d(\cdot,\cdot)$ is for instance the Mahalanobis distance defined through the nonconstant auxiliary variables
\begin{align*}
 d( i, j ) = \left\{ \left(\xb_i-\xb_j\right)^\top \Sigmab^{-1} \left(\xb_i-\xb_j\right) \right\}^{1/2},
\end{align*}
where $\Sigmab$ is the variance-covariance matrix of the nonconstant auxiliary variables. If the values of the auxiliary variables are known on the sample level only, $\Sigmab$ must be estimated.

The random $k$-nearest neighbor imputation method ($k$NNI) consists of replacing the missing value of a unit $j \in S_m$ with the value of one of its $k$-nearest neighbors. The donors are randomly selected with equal probabilities. This results in the matrix of imputation probabilities $\psib^{\knn} = \left( \psi_{ij}^{\knn} \right)$, $(i,j) \in S_r \times S_m$ with
\begin{equation}\label{equation:psi:knn}
\psi_{ij}^{\knn} =
    \left\{
        \begin{array}{ll}
        \frac{1}{k} & \hbox{if $i$ belongs to the $k$ nearest neighbors of $j$,}\\
                0   & \hbox{otherwise.}
        \end{array}
    \right.
\end{equation}
Notation $\knn$ in superscript of the matrix $\psib$ means that the matrix linked to $k$NNI is considered. The matrix of imputation probabilities related to the $k$NNI is a matrix containing exactly $k$ non-null coefficients in each column and all these non-null coefficients are equal to $1/k$.
This particular matrix of imputation probabilities is the starting point of the method proposed in this paper.

\section{Balanced \texorpdfstring{$k$}{k}-nearest neighbor imputation method}\label{section:balanced:knn}

\subsection{Aim of the method}\label{subsection:balanced:knn}

Random imputation methods show the nice feature that they tend to preserve the distribution of the variable being imputed. This is often not the case for deterministic imputation methods. In return, this randomness implies the undesirable presence of imputation variance. Hence, a random imputation method that makes it possible to keep the imputation variance relatively small is of great interest. Moreover, donor imputation methods, such as hot-deck, have the advantage of imputing feasible and observed values.

An imputation method can be based on either a parametric procedure or a nonparametric procedure. The hypotheses underlain by parametric procedures are strong. In this case, one can assume that the distribution of the variable of interest is known, but the parameters of this distribution must be estimated. Practically, this kind of hypotheses is satisfied in only a few cases.
On the other hand, hypotheses underlain by nonparametric procedures are much weaker. Hence, such procedures are usually much more flexible than parametric procedures.
An imputation method based on a nonparametric procedure therefore makes it possible to handle a wider range of data without violating its underlain hypotheses than an imputation method based on a parametric procedure.

The proposed method shows the nice features stated above. First, it is a random hot-deck imputation method. It therefore tends to preserve the distribution of the variable being imputed and it imputes observed and feasible values. Moreover, even though it is random, the proposed method makes it possible to control the imputation variance of the total estimator. Indeed, the imputation process conserves the estimator of the total of the auxiliary variables. If the auxiliary variables suffered from nonresponse, their imputed total estimator would match their total estimator under complete response. The imputation mechanism relative to b$k$NNI is therefore such that conditionally on the sampling mechanism and on the nonresponse mechanism
\begin{equation}
     \widehat{\Xb}_I  =  \widehat{\Xb}.       \label{equation:imputation:error:variance}
\end{equation}
If a linear model fits well the relation between the variable of interest and the auxiliary variables, the imputation variance of the total estimator can in this way be reduced. As a result, the proposed method makes it possible to impute randomly while keeping imputation variance of the total estimator relatively small. Then, the proposed method is based on a nonparametric procedure. Indeed, the donors are chosen in neighborhoods of recipients. For each nonrespondent, a donor is randomly selected among its $k$-nearest neighbors. As nonparametric, this procedure is very adaptive and makes it possible to impute values that are close to the unobserved missing value for a wide range of data. Moreover and as stated above, the imputation process conserves the estimator of the total of the auxiliary variables. If a linear model fits well the relation between the variable of interest and the auxiliary variables, this property implies that the imputed total estimator for the variable being imputed is close to the total estimator that would have been obtained under complete response (Horvitz-Thompson estimator). As this last one is unbiased, the obtained imputed total estimator is nearly unbiased in the case specified above. Details regarding the properties of the total estimator are given in Section~\ref{section:model}. Last, the proposed method used the methodology proposed in Section~\ref{section:methodology:hotdeck}. This makes it possible for the user to take into account edit rules directly while imputing as explained in Section~\ref{section:methodology:hotdeck}.

In what follows, it is explained how donors can be obtained. Notation $\bknn$ in the superscript of the matrices $\psib$ and $\phib$ means that the matrices linked to the b$k$NNI are considered. The method proceeds in two steps. The first step consists of obtaining the matrix of imputation probabilities $\psib^{\bknn}$ whereas the second step consists of generating a realization of the matrix of imputation $\phib^{\bknn}$. The aim is that the imputation mechanism satisfies Equation~\eqref{equation:imputation:error:variance}. With this aim, the matrix of imputation probabilities $\psib^{\bknn}$ is, in the first step, constructed such that Equation
\begin{equation}\label{equation:constraint:psi:0}
    \E_I\left(\widehat{\Xb}_I\right) = \widehat{\Xb}
\end{equation}
is satisfied and the matrix of imputation probabilities $\phib^{\bknn}$ is, in the second step, generated such that Equation
\begin{equation}\label{equation:constraint:phi:0}
    \widehat{\Xb}_I = \E_I\left(\widehat{\Xb}_I\right) .
\end{equation}
is satisfied. These two steps are presented in Section~\ref{subsection:imputation:probabilities} and in Section~\ref{subsection:selection:donors} respectively. Equation~\eqref{equation:constraint:psi:0} together with Equation~\eqref{equation:constraint:phi:0} lead to Equation~\eqref{equation:imputation:error:variance}. As a result, the imputation mechanism obtained through the two steps described above satisfies Equation~\eqref{equation:imputation:error:variance}.

\subsection{Calibration}\label{section:calibration}

The aim of this Section is to briefly describe calibration \citep{dev:sar:92} which is the main tool used in Section~\ref{subsection:imputation:probabilities} to obtain the matrix of imputation probabilities $\psib^{\bknn}$. Suppose a vector of $Q$ auxiliary variables $\xb_i = \left( x_{i1}, x_{i2}, \ldots, x_{iQ}\right)^\top$ is known for each unit of the population $U$. The aim of calibration is to find calibration weights $w_i$ for $i \in S$ as close as possible to the initial design weights $d_i = 1 / \pi_i$ (in an average sense for a given distance) while respecting the calibration equation
\begin{align*}
    \sum_{i \in S} w_i \xb_i = \sum_{i \in U} \xb_i = \Xb.
\end{align*}
Several distance functions are proposed in~\cite{dev:sar:92} as a means of measuring the distance between the initial design weights $d_i = 1 / \pi_i$ and the final weights $w_i$. Each distance provides a particular form for the final weights $w_i$. The raking method is the calibration obtained considering the distance function $G(\cdot,\cdot)$ given by
\begin{align*}
    G(w_i,d_i) = w_i \log \left( \frac{w_i}{d_i} \right) - w_i + d_i.
\end{align*}
This leads to the final weights $ w_i = d_i \exp \left( \lambdab^\top \xb_i   \right) $ where the vector $\lambdab = \left( \lambda_1, \lambda_2, \ldots, \lambda_Q  \right)^\top$ is the solution to the calibration equation
\begin{align*}
    \sum_{i \in S} d_i \exp \left( \lambdab^\top \xb_i   \right) \xb_i = \sum_{i \in U} \xb_i.
\end{align*}
From~\cite{dev:sar:92}, if a solution $\lambdab$ to this calibration problem exists, then it is unique.

\subsection{Obtaining the matrix of imputation probabilities \texorpdfstring{$\psib^{\bknn}$}{}}\label{subsection:imputation:probabilities}

In this Section, a procedure to obtain the matrix of imputation probabilities $\psib^{\bknn} $ is presented. This matrix must satisfy
\begin{equation}
    \psi_{ij}^{\bknn} \neq 0    \quad   \mbox{only if}  \quad    i \in \mbox{knn}(j), \label{equation:neighborhood:principle}
\end{equation}
because it is imposed that donors are chosen in neighborhoods of recipients. Moreover, as stated in Section~\ref{section:methodology:hotdeck}, a matrix of imputation must satisfy equations~\eqref{equation:constraint:P:sum1} and \eqref{equation:constraint:P:coeff01}. Finally, as stated in Section~\ref{subsection:balanced:knn}, in order to select donors such that Equation~\eqref{equation:imputation:error:variance} is satisfied, it is imposed on $\psib^{\bknn}$ to satisfy Equation~\eqref{equation:constraint:psi:0}, i.e
\begin{equation}\label{equation:constraint:psi}
    \E_I\left(\widehat{\Xb}_I\right) = \widehat{\Xb}.
\end{equation}
This constraint states that conditionally on the sampling mechanism and the nonresponse mechanism, the imputation mechanism provides an unbiased imputed total estimator for the auxiliary variables.
Equation~\eqref{equation:constraint:psi} is equivalent to
\begin{align}\label{equation:constraint:psi:2}
    \sum_{ j \in S_m } d_j \sum_{i \in S_r} \psi_{ij}^{\bknn} \xb_i     &=     \sum_{j \in S_m} d_j \xb_j.
\end{align}
Hence, the matrix of imputation probabilities $\psib^{\bknn} = \left( \psi_{ij}^{\bknn} \right) $ must satisfy simultaneously
\begin{align}
    \psi_{ij}^{\bknn}                   &\neq 0    \quad   \mbox{only if}  \quad    i \in \mbox{knn}(j),    \label{equation:neighborhood:principle:2}\\
    \sum_{i \in S_r} \psi_{ij}^{\bknn}  &= 1 \quad \mbox{for each } j \in S_m,                     \label{equation:constraint:P:sum1:2}\\
    \psi_{ij}^{\bknn}            &\geq 0 \quad \mbox{for each } (i,j) \in S_r \times S_m,    \label{equation:constraint:P:coeff01:2}\\
    \sum_{ j \in S_m } d_j \sum_{i \in S_r} \psi_{ij}^{\bknn} \xb_i     &=     \sum_{j \in S_m} d_j \xb_j.       \label{equation:constraint:psi:3}
\end{align}

The existence of a matrix of imputation probabilities $\psib^{\bknn}$ satisfying the above conditions and the choice of the number of nearest neighbors $k$ are discussed in Section~\ref{section:choice:k}. Algorithm~\ref{algorithm:matrix:imp:prob} presents a procedure to obtain the matrix of imputation probabilities $\psib^{\bknn}$. The main idea of Algorithm~\ref{algorithm:matrix:imp:prob} is to find a matrix of imputation probabilities $\psib^{\bknn}$ close to the matrix of imputation probabilities relative to the $k$NNI $\psib^{\knn}$, while satisfying equations~\eqref{equation:neighborhood:principle:2}, \eqref{equation:constraint:P:sum1:2}, \eqref{equation:constraint:P:coeff01:2}, and \eqref{equation:constraint:psi:3}. This Algorithm initializes with the matrix $\psib^{\knn}$. Throughout all the steps, a null coefficient remains null which implies that equation~\eqref{equation:neighborhood:principle:2} is satisfied. Thereafter, calibrations and normalizations are alternated. Calibrations provide matrices $\psib(2\ell)$ for $\ell \geq 1$ satisfying equations~\eqref{equation:constraint:P:coeff01:2} and \eqref{equation:constraint:psi:3}. However, these matrices $\psib(2\ell)$ are not matrices of imputation probabilities because they do not satisfy~\eqref{equation:constraint:P:sum1:2}. Normalizations provide matrices $\psib(2\ell+1)$ for $\ell \geq 1$ satisfying equations~\eqref{equation:constraint:P:sum1:2} and \eqref{equation:constraint:P:coeff01:2} but not necessarily satisfying equation~\eqref{equation:constraint:psi:3}. Iterations stop when the matrix $\psib(2\ell+1)$ obtained by normalization approximately satisfies equation~\eqref{equation:constraint:psi:3}. The matrix $\psib^{\bknn}$ is the last $\psib(2\ell+1)$ considered. Hence, $\psib^{\bknn}$ satisfies equations~\eqref{equation:neighborhood:principle:2}, \eqref{equation:constraint:P:sum1:2}, \eqref{equation:constraint:P:coeff01:2}, and \eqref{equation:constraint:psi:3} simultaneously.

\begin{algorithm}\small
\caption{Procedure to obtain the matrix of imputation probabilities $\psib^{\bknn}$}
\label{algorithm:matrix:imp:prob}
\begin{itemize}
\item   \emph{Initialization}\\
        Set $\psib(1) = \psib^{\knn}$, the matrix of imputation probabilities relative to the $k$NNI defined in Expression~\eqref{equation:psi:knn}.
\item   \emph{Iterations}\\
        Repeat for $\ell = 1,2,\ldots$
        \begin{itemize}
            \item   \emph{Calibration}\\
                    For $i \in S_r$, consider the initial weights $\tilde{d}_i = \sum_{j \in S_m} d_j \psi(2\ell - 1)_{ij}$ and obtain the calibrated weights $w_i =  \tilde{d}_i  \exp\left( \lambdab^\top \xb_i \right)$ by means of the raking method. The calibration equation is
                    \begin{align*}
                        \sum_{i \in S_r} w_i \xb_i = \sum_{j \in S_m} d_j \xb_j.
                    \end{align*}
            \item   Let $\psib(2\ell)$ be the matrix defined as $\psi(2\ell)_{ij} = \psi(2\ell-1)_{ij} \exp\left( \lambdab^\top \xb_i \right)$.
            \item   \emph{Normalization}\\
                    Let $\psib(2\ell+1)$ be the matrix defined as $\psi(2\ell+1)_{ij} = \frac{ \psi(2\ell)_{ij} }{ \sum_{l \in S_r} \psi(2\ell)_{lj} } $.
            \item   \emph{Stop criterion}\\
                    If
                    \begin{align*}
                        \max_{1 \leq q \leq Q}
                        \left|
                            \frac   { \sum_{ j \in S_m } \sum_{i \in S_r} d_j \psi(2\ell+1)_{ij} x_{iq}  -   \sum_{j \in S_m} d_j x_{jq} }
                                    {   \sum_{j \in S_m} d_j x_{jq} }
                        \right| \leq tol
                    \end{align*}
                    for a small fixed error tolerance $tol$, then
                    \begin{itemize}
                        \item   Set $\psib^{\bknn} = \psib(2\ell+1)$,
                        \item   Stop.
                    \end{itemize}
        \end{itemize}
 \end{itemize}
 \end{algorithm}

\subsection{Choice of \texorpdfstring{$k$}{k} and existence of \texorpdfstring{$\psib^{\bknn}$}{the matrix of imputation probabilities}}\label{section:choice:k}

In this Section, the choice of the number of nearest neighbors $k$ and existence of a matrix of imputation probabilities $\psib^{\bknn}$ having the required properties are discussed.

The number of nearest neighbors $k$ is relevant when constructing the matrix $\psib^{\bknn}$ in Section~\ref{subsection:imputation:probabilities}. Indeed, this matrix must contain at most $k \cdot n_m$ non-null coefficients as confirmed by Equation~\eqref{equation:neighborhood:principle:2}. Moreover, the coefficients of this matrix must satisfy Equation~\eqref{equation:constraint:P:sum1:2} and Equation~\eqref{equation:constraint:psi:3}. These two equations together form a system of $n_m + q$ constraints. Hence, when constructing the matrix of imputation probabilities $\psib^{\bknn}$, the aim is to find $k \cdot n_m$ unknown coefficients that satisfy $n_m + q$ linear constraints. As a result, a necessary condition to find a matrix $\psib^{\bknn}$ satisfying Equations~\eqref{equation:neighborhood:principle:2},~\eqref{equation:constraint:P:sum1:2} and~\eqref{equation:constraint:psi:3} is
\begin{equation}\label{equation:constraint:k}
    k \geq \frac{n_m + q}{n_m}.
\end{equation}
However, this condition is not sufficient to ensure that a solution $\psib^{\bknn}$ exists and elements external to the choice of $k$ have an effect. Such an element is the configuration of the nonrespondents. Indeed, Equation~\eqref{equation:constraint:P:coeff01:2} and Equation~\eqref{equation:constraint:P:sum1:2} imply together that all the coefficients $\psi_{ij}^{\bknn}$ must lie between 0 and 1. Consider the extreme case in which all the units with the largest values of one auxiliary variable are nonrespondents. For this auxiliary variable, and as all the coefficients $\psi_{ij}^{\bknn}$ must lie between 0 and 1, the small values of respondents do not make it possible to compensate the large missing values of the nonrespondents. Hence, a matrix $\psib^{\bknn}$ with the properties stated above might not exist, whatever the value of $k$ is. Condition~\eqref{equation:constraint:k} is therefore necessary but not sufficient for a solution to exist. Note that, if a solution exists, then it is unique. Indeed, the solution to the calibration problem in Algorithm~\ref{algorithm:matrix:imp:prob} is unique (see last sentence of Section~\ref{section:calibration}).

The choice of $k$ can have an impact on the bias of the total estimator. Indeed, the smaller $k$ is, the closer the values of the auxiliary variables are in a neighborhood of $k$ units and therefore the closer the values of the variable of interest $y$ tend to be in the same neighborhood. As a result and from Property~\ref{thm:neighborhood:principle} below, under the hypothesis that the values of variables $y$ are similar in a neighborhood,  the smaller $k$ is, the smaller the bias of the total estimator tends to be. Hence, with the aim of controlling the bias, $k$ should be chosen as small as possible. However, note that this choice has no impact on the bias of the total estimator when the hypotheses of Property~\ref{thm:linear:model} or those of Property~\ref{thm:response:model} are satisfied. Indeed, the bias of the total estimator are in these cases null whatever the value of $k$ is. Moreover, the larger $k$ is, the larger the imputation variance of the total is.

For these reasons, we suggest to choose the smallest value $k$ for which a solution for the computation of matrix $\psib^{\bknn}$ exists. Thus, we propose to fix the value of $k$ as follows. First, choose a value of $k$ that satisfies Equation~\eqref{equation:constraint:k} and that is relatively small. Then, apply Algorithm~\ref{algorithm:matrix:imp:prob} and see if this one finds a solution, i.e. returns a matrix $\psib^{\bknn}$ satisfying the required conditions. If this is the case, the user can then apply Algorithm~\ref{algorithm:matrix:imp} in order to select the donors. Otherwise, we propose gradually increasing the value of $k$ and repeating this procedure until a solution is found.

\subsection{Stratified balanced sampling}

The aim of this Section is to briefly describe stratified balanced sampling. This represents the main tool used in Section~\ref{subsection:selection:donors} to obtain the matrix of imputation $\phib^{\bknn}$. Suppose a vector of $Q$ auxiliary variables $\xb_i = \left( x_{i1}, x_{i2}, \ldots, x_{iQ}\right)^\top$ is known for each unit of the population $U$. A sampling design $p(\cdot)$ is said to be balanced on these auxiliary variables if
\begin{align*}
    \sum_{i \in s} \frac{ \xb_i }{ \pi_i } = \sum_{i \in U} \xb_i,
\end{align*}
for each $s \subset U$ with $p(s)>0$. This last equation can be rewritten
\begin{align*}
    \sum_{i \in U} \frac{ \xb_i }{ \pi_i } \mathbbm{1}_{i \in s} = \sum_{i \in U} \xb_i,
\end{align*}
where $\mathbbm{1}_{i \in s}$ is the indicator function which takes value 1 if unit $i$ belongs to $s$ and 0 otherwise. The cube method \citep{dev:til:04a} is a method for balanced sampling. It allows a balanced sample to be selected while satisfying the inclusion probabilities in the sense that $\E \left( \mathbbm{1}_{i \in S} \right) = \pi_i $ for each $i \in U$. Suppose moreover that the population $U$ is partitioned into $H$ nonoverlapping strata $U_1, U_2,\ldots,U_H$. A stratified balanced sampling design is a sampling design balanced in each stratum, i.e.
\begin{equation}\label{equation:stratified:balanced:sampling:1}
    \sum_{i \in U_h} \frac{ \xb_i }{ \pi_i } \mathbbm{1}_{i \in s} = \sum_{i \in U_h} \xb_i     \quad   \mbox{for each } 1 \leq h \leq H,
\end{equation}
for each $s \subset U$ with $p(s) > 0$. \cite{cha:09} and~\cite{has:til:14} proposed methods for stratified balanced sampling which are based on the cube method. The algorithm proposed in~\cite{has:til:14} is particularly fast and is applicable when the number of strata is very large. The samples selected with these methods are approximately balanced in each stratum and approximately balanced in the overall population while satisfying the inclusion probabilities in the sense that
\begin{equation}\label{equation:stratified:balanced:sampling:2}
\E \left( \mathbbm{1}_{i \in S} \right) = \pi_i \quad \mbox{for each } i \in U.
\end{equation}
\begin{rem}\label{rem:sum:pi:sbs}
If the sum of the inclusion probabilities
\begin{align*}
    n_h = \sum_{i \in U_h} \pi_i
\end{align*}
is integer in each stratum $U_h$, $1 \leq h \leq H$, the algorithm proposed in~\cite{has:til:14} selects exactly $n_h$ units in each stratum $U_h$, $1 \leq h \leq H$.
\end{rem}

\subsection{Selection of the donors}\label{subsection:selection:donors}

In this Section, a procedure to obtain the matrix of imputation $\phib^{\bknn} $ from the matrix of imputation probabilities $\psib^{\bknn}$ is presented. The main idea of this procedure is to select a donor among the respondents for each nonrespondent while satisfying constraints. The key feature of the method is that the selection of donors is viewed as a sampling problem and the constraints are satisfied through stratified balanced sampling, where only one unit is selected in each stratum.

As stated in Section~\ref{subsection:balanced:knn}, in order to select donors such that Equation~\eqref{equation:imputation:error:variance} is satisfied, it is imposed on $\phib^{\bknn}$ to satisfy Equation~\eqref{equation:constraint:phi:0}, i.e
\begin{equation}
    \widehat{\Xb}_I = \E_I\left(\widehat{\Xb}_I\right) \label{equation:constraint:phi}.
\end{equation}
This constraint states that the imputation variance of the auxiliary variables cancels out. A sufficient condition for Equation~\eqref{equation:constraint:phi} to hold is
\begin{align*}
    \sum_{i \in S_r} d_j \phi_{ij}^{\bknn} \xb_i   & = \sum_{i \in S_r} d_j \psi_{ij}^{\bknn} \xb_i \quad \mbox{for each } j \in S_m,
\end{align*}
which can be rewritten
\begin{equation} \label{equation:constraint:phi:2}
    \sum_{i \in S_r}\frac{ d_j \psi_{ij}^{\bknn} \xb_i  }{\psi_{ij}^{\bknn}} \phi_{ij}^{\bknn}  =   \sum_{i \in S_r}  d_j \psi_{ij}^{\bknn} \xb_i   \quad   \mbox{for each } j \in S_m.
\end{equation}
Moreover, matrices $\psib^{\bknn}$ and $\phib^{\bknn} $ must be linked by Equation~\eqref{equation:link:psi:phi} , i.e.
\begin{equation}\label{equation:link:psi:phi:1}
\psi^{\bknn}_{ij} = \E_I\left(\phi^{\bknn}_{ij}\right)\mbox{for each } (i,j) \in S_r \times S_m.
\end{equation}
The aim is therefore to generate matrix $\phib^{\bknn}$ such that
\begin{align}
    \sum_{i \in S_r}\frac{ d_j \psi_{ij}^{\bknn} \xb_i  }{\psi_{ij}^{\bknn}} \phi_{ij}^{\bknn}  &=   \sum_{i \in S_r}  d_j \psi_{ij}^{\bknn} \xb_i   \quad   \mbox{for each } j \in S_m,             \label{equation:constraint:phi:3}\\
    \E_I \left( \phi_{ij}^{\bknn} \right) &= \psi_{ij}^{\bknn} \quad \mbox{for each }  (i,j) \in S_r \times S_m   \label{equation:link:psi:phi:2}.
\end{align}
A matrix $\phib^{\bknn}$ satisfying exactly the above equations often does not exist. However, when generating this matrix with the procedure presented below, the constraints are relaxed until a solution is found. See Remark~\ref{rem:2}. A solution to this problem was proposed in~\cite{cha:dev:haz:11}. A slight modification of that is presented here. Consider the population of cells
\begin{equation}
    \dot{U} = \left\{ (i,j) | i \in S_r, j \in S_m\right\}.
\end{equation}
This population is partitioned into $n_m$ strata $\dot{U}_j$, $j \in S_m$ where $\dot{U}_j = \left\{ (i,j) | i \in S_r\right\}$. Each stratum corresponds to one nonrespondent. Then, exactly one unit will be selected in each stratum, providing in this way exactly one donor for each nonrespondent. For each unit $(i,j) \in \dot{U}$, consider the initial inclusion probability
\begin{align*}
    \dot{\pi}_{(i,j)} =  \psi_{ij}^{\bknn},
\end{align*}
and the auxiliary variables
\begin{align*}
    \dot{\xb}_{(i,j)} = d_j \psi_{ij}^{\bknn} \xb_i.
\end{align*}
Moreover, as $\phi_{ij}^{\bknn}$ is 1 if respondent $i$ is the donor for nonrespondent $j$ and 0 otherwise, consider
\begin{align*}
    \mathbbm{1}_{(i,j) \in S} = \phi_{ij}^{\bknn}.
\end{align*}
It means that respondent $i \in S_r$ is used to impute the missing value of nonrespondent $j \in S_m$ if unit $(i,j)$ is selected in the sample. The problem formed by equations~\eqref{equation:constraint:phi:3} and \eqref{equation:link:psi:phi:2} can be rewritten as follows
\begin{align}
    \sum_{(i,j) \in \dot{U}_j} \frac{ \dot{\xb}_{(i,j)}  }{ \dot{\pi}_{(i,j)}} \mathbbm{1}_{(i,j) \in S}  &=   \sum_{(i,j) \in \dot{U}_j}  \dot{\xb}_{(i,j)}   \quad   \mbox{for each } j \in S_m  ,             \label{equation:constraint:phi:4}\\
    \E_I \left( \mathbbm{1}_{(i,j) \in S} \right) &= \dot{\pi}_{(i,j)} \quad \mbox{for each }  (i,j) \in \dot{U}   \label{equation:link:psi:phi:3}.
\end{align}
This is a typical problem of stratified balanced sampling where only one unit is selected in each stratum because each respondent receives exactly one value. Equations~\eqref{equation:constraint:phi:4} and \eqref{equation:link:psi:phi:3} correspond respectively to equations~\eqref{equation:stratified:balanced:sampling:1} and \eqref{equation:stratified:balanced:sampling:2}. The procedure to obtain matrix $\phib^{\bknn}$ therefore uses stratified balanced sampling and is presented in Algorithm~\ref{algorithm:matrix:imp}. The first step of the algorithm consists of selecting a stratified balanced sample in the cell population $\dot{U} = \left\{ (i,j) | i \in S_r, j \in S_m\right\}$ such that equations~\eqref{equation:constraint:phi:4} and \eqref{equation:link:psi:phi:3} are satisfied. In a second and last step, matrix $\phib^{\bknn}$ is obtained by setting $\phi^{\bknn}_{ij}=1$ if the cell $(i,j)$ has been selected in the sample and 0 otherwise.

\begin{rem}
    For each $j \in S_m$, the following equation is satisfied
    \begin{align*}
        \sum_{(i,j) \in \dot{U}_j} \dot{\pi}_{(i,j)} = \sum_{i \in S_r}  \psi_{ij}^{\bknn} = 1.
    \end{align*}
     This means that the sum of the inclusion probabilities is equal to 1 in each stratum considered in the stratified balanced sampling problem of Algorithm~\ref{algorithm:matrix:imp}. Therefore, as the procedure for balanced stratified sampling proposed in~\cite{has:til:14} is applied in Algorithm~\ref{algorithm:matrix:imp} and from Remark~\ref{rem:sum:pi:sbs}, exactly 1 unit is selected in each stratum, i.e
    \begin{align*}
        \sum_{ (i,j) \in \dot{U}_j } \mathbbm{1}_{(i,j) \in S} = 1 \quad    \mbox{for each } j \in S_m.
    \end{align*}
    Moreover, as
    \begin{align*}
        \sum_{ (i,j) \in \dot{U}_j } \mathbbm{1}_{(i,j) \in S} = \sum_{i \in S_r} \phi^{\bknn}_{ij}\quad \mbox{for each }j \in S_m,
    \end{align*}
    the matrix $\phib^{\bknn}$ satisfies Equation~\eqref{equation:constraint:sum:imputation}. This means that
    Algorithm~\ref{algorithm:matrix:imp} provides exactly one donor for each nonrespondent $j \in S_m$.
\end{rem}

\begin{rem}\label{rem:2}
It is often not possible to select samples such that the balancing equations are exactly satisfied. As a result, Algorithm~\ref{algorithm:matrix:imp} often generates a matrix $\phib^{\bknn}$ such that Equation~\eqref{equation:constraint:phi:4} is only approximately satisfied. Equivalently, donors are often selected such that Equation~\eqref{equation:constraint:phi} is only approximately satisfied.
\end{rem}

\begin{algorithm}\small
\caption{Procedure to obtain the matrix of imputation $\phib^{\bknn}$}
\label{algorithm:matrix:imp}
\begin{itemize}
    \item   \emph{Stratified balanced sampling}\\
            Select a stratified balanced sample $S$ in the cells population $\dot{U} = \left\{ (i,j) | i \in S_r, j \in S_m\right\}$ with the method proposed in~\cite{has:til:14} such that
            \begin{align*}
                \sum_{(i,j) \in \dot{U}_j} \frac{ \dot{\xb}_{(i,j)}  }{ \dot{\pi}_{(i,j)}} \mathbbm{1}_{(i,j) \in S}  &=   \sum_{(i,j) \in \dot{U}_j}  \dot{\xb}_{(i,j)}   \quad   \mbox{for each } j \in S_m,\\
                \E_I \left( \mathbbm{1}_{(i,j) \in S} \right) &= \dot{\pi}_{(i,j)} \quad \mbox{for each }  (i,j) \in \dot{U},
            \end{align*}
            where
            \begin{itemize}
                \item   $\dot{\xb}_{(i,j)} = d_j \psi_{ij}^{\bknn} \xb_i$ is the vector of balancing variables linked to unit $(i,j) \in S_r \times S_m$,
                \item   $\dot{\pi}_{(i,j)} =  \psi_{ij}^{\bknn}$ is the inclusion probability attached to unit $(i,j) \in S_r \times S_m$,
                \item   $\dot{U}_j = \left\{ (i,j) | i \in S_r\right\}$ is the stratum attached to unit $j \in S_m$.
            \end{itemize}
    \item   \emph{Matrix of imputation}\\
            Let $\phib^{\bknn}$ be the matrix defined as $\phi^{\bknn}_{ij} = \mathbbm{1}_{(i,j) \in S}$.
 \end{itemize}
 \end{algorithm}

\section{Approximation of conditional imputation variance}\label{section:estimation:imputation:variance}

A procedure to approximate the imputation variance conditional to the design and on the nonresponse mechanism of the total $\Var_{I}(\widehat{Y}_I)$ is described in this Section. For the sake of brevity, we refer to the latter variance as conditional imputation variance, because it is conditional to the sampling design and the nonresponse mechanism. The proposed procedure relies on the idea that, in the new method, the selection of donors is viewed as a sampling problem and imputation is achieved through stratified balanced sampling. However, all the values of the variable of interest are known prior to selecting the sample of imputed values. Indeed, donors are selected among respondents and the values of the variable of interest are fully observed for these. \cite{dev:til:05} proposed an approximation formula for the variance of a total under balanced sampling that can be used in this framework. Based on this approximation formula, we propose the following formula for the conditional imputation variance.
\begin{align}\label{var:imp:approx}
    \Var^{app}_{I}(\widehat{Y}_I) = \sum_{i \in S_r} \sum_{ \substack{j \in S_m \\ \psi_{ij}^{\bknn} \neq 0}}
                                            c_{ij} d_j^2 \left( y_i - \bb^\top \xb_i   \right)^2,
\end{align}
\begin{align}
    c_{ij}  &= \psi_{ij}^{\bknn} \left( 1 - \psi_{ij}^{\bknn} \right)
                \frac{ n_m k } { n_m k - q},\\
    \bb  &= \left( \sum_{i \in S_r} \sum_{ \substack{j \in S_m \\ \psi_{ij}^{\bknn} \neq 0}}  c_{ij} d_j^2 \xb_i \xb_i^\top  \right)^{-1}
               \sum_{i \in S_r} \sum_{ \substack{j \in S_m \\ \psi_{ij}^{\bknn} \neq 0}} c_{ij} d_j^2 \xb_i  y_i.
\end{align}

Notice that only a single respondents set is necessary to approximate the conditional imputation variance. However, it can underestimate this one. Indeed, this formula comes from the variance of the total for balanced sampling. When applying balanced sampling, it is often not possible to exactly satisfy the balancing constraints, which is referred to as a \emph{rounding problem}. A sample can thus be only approximately balanced \citep[see][]{dev:til:04a}.
Indeed, the balancing constraints must be relaxed in order to make it possible to select a sample. Hence, the variance of the total when balanced sampling is applied can be broken down into two terms: a first term derived under the hypothesis that the balancing equations are perfectly satisfied and a second term due to the rounding problem \citep[see][]{dev:til:05}. As the other approximations and estimators of the variance proposed in this framework, Formula~\eqref{var:imp:approx} does not capture the part of the variance due to the rounding problem and can therefore underestimate the conditional imputation variance.

The stronger the linear relation between the variable of interest and the auxiliary variables is, the more formula~\eqref{var:imp:approx} tends to underestimates the conditional imputation variance. To understand the reason for this, suppose that there is a strict linear relation between the variable of interest and the auxiliary variables. In this case, if the auxiliary variables are perfectly balanced, so is also the variable of interest. As a result, the term in the variance due to the balancing itself is null. Hence, the variance is only due to the rounding problem. In this case, an estimator that captures only the term due to the balancing therefore captures 0\% of the actual variance. Then, as the linear relation between the variable of interest and the auxiliary variables weakens, the variable of interest becomes less well balanced. This implies that the variance due to the balancing increases, and, therefore, that the part of the actual variance that is returned by an estimator that captures only the variance due to this one increases.

\section{Properties of the imputed total estimator}\label{section:model}

The performance of the imputed total estimator under the proposed imputation method relies on two underlying models: the linear model and the response model. Moreover, as donors are chosen in neighborhoods of recipients, the performance of the imputed total estimator depends on a third principle, which is the neighborhood principle. These models and principles are here detailed. Moreover, the asymptotic properties of the imputed total estimator are studied.

\subsection{Linear model}\label{section:linear:model}

\begin{thm}\label{thm:linear:model}
\begin{enumerate}
    \item   Suppose the data is MAR (or MCAR). Consider the linear model
            \begin{align*}
                \mbox{m:} \quad y_i = \betab ^\top \xb_i + \varepsilon_i \quad \mbox{with} \quad \E_m\left( \varepsilon_i \right) = 0,
            \end{align*}
            where $\E_m (.)$ denotes the expectation with respect to the model m. If the model m holds, then the b$k$NNI provides an unbiased imputed total estimator $\widehat{Y}_I$ in the sense that
            \begin{align*}
                \Bias \left( \widehat{Y}_I \right) = \E_m \E_p \E_q \E_I \left( \widehat{Y}_I - Y \right) = 0.
            \end{align*}
    \item   Moreover, if the relation between the variable of interest and the auxiliary variables is strictly linear, i.e.
            \begin{align*}
                y_i = \betab ^\top \xb_i,
            \end{align*}
            then the b$k$NNI provides an imputed total estimator $\widehat{Y}_I$ with a quasi-null imputation variance, i.e.
            \begin{align*}
            \Var_{imp} = \E_p \E_q \Var_I  \left( \widehat{Y}_I \right) \approx 0 ,
            \end{align*}
            regardless of the nonresponse process (MCAR, MAR or NMAR).
\end{enumerate}
\end{thm}

The proof is given in the Appendix. It results that if the linear model $m$ reasonably fits the population data, then the b$k$NNI provides an almost unbiased imputed total estimator $\widehat{Y}_I$ with a small imputation variance.

\subsection{Response model}\label{section:response:model}

\begin{thm}\label{thm:response:model}
     Let $\psib^{\bknn} = \left( \psi_{ij}^{\bknn} \right), (i,j) \in S_r \times S_m$, be the matrix of imputation probabilities relative to the b$k$NNI. If
   \begin{align}
        \theta_i = \frac{ 1 }{ 1 + \sum_{j \in S_m} \frac{d_j}{d_i} \psi_{ij}^{\bknn}}, \label{estimated:response:probabilities}
   \end{align}
   perfectly fits the true response probability of each unit $i \in S_r$, then the b$k$NNI provides an unbiased imputed total estimator $\widehat{Y}_I$, i.e.
   \begin{align*}
        \Bias \left( \widehat{Y}_I \right) = \E_p \E_q \E_I \left( \widehat{Y}_I - Y \right) = 0.
   \end{align*}
\end{thm}

The proof is given in the Appendix. It results that if the model stated above estimates the response probabilities for $i \in S_r$ reasonably well, then the imputed estimator $\widehat{Y}_I$ is an almost unbiased estimator for $Y$.

This result can be interpreted in the following way. Respondent $i \in S_r$ acts as a donor a certain number of times in such a way that in expectation its weight is equal to
\begin{align}
    d_i + \sum_{j \in S_m} d_j \psi_{ij}^{\bknn}.
\end{align}
If $\theta_i$ is the response probability of unit $i \in S_r$, then the bias due to nonresponse is compensated. This can happen when the auxiliary variables $\xb_i$ can describe the nonresponse mechanism, because the weights $\psi_{ij}^{\bknn}$ are obtained by calibration on the estimated totals of these variables.

\subsection{Neighborhood principle}\label{section:neighborhood}

If none of the two previous models hold, a third principle can correct the situation, namely the neighborhood principle. The neighborhood principle states that neighboring units (i.e units showing close auxiliary values) show close $y$ values.

\begin{thm}\label{thm:neighborhood:principle}
    Consider $(i,j) \in S_r \times S_m$. If the implication $$i \in knn(j) \quad \Rightarrow \quad  y_i - y_j  = 0$$ holds, then the b$k$NNI provides an unbiased imputed total estimator $\widehat{Y}_I$, i.e.
    \begin{align*}
        \Bias \left( \widehat{Y}_I \right) = \E_p \E_q \E_I \left( \widehat{Y}_I - Y \right) = 0.
    \end{align*}
\end{thm}

The proof is given in the Appendix. It results that if neighboring units (i.e units showing close auxiliary values) have $y$ values which are close, then the imputed estimator $\widehat{Y}_I$ is an almost unbiased estimator for $Y$.

\subsection{Resistance to model misspecification}\label{section:robustness}

The new method is resistant to model misspecification in terms of bias of the imputed total estimator $\widehat{Y}_I$. It is indeed sufficient that only one of the three models or principle stated above holds to obtain an unbiased imputed total estimator $\widehat{Y}_I$. However, a unique model provides an imputed total estimator $\widehat{Y}_I$ with a quasi-null imputation variance, namely the strictly linear model
$    y_i = \betab ^\top \xb_i$.

\subsection{Asymptotic properties of the total estimator}\label{section:asymptotic}
We now study the asymptotic properties of the total estimator under b$k$NNI. Suppose that there is a sequence of finite populations indexed by $\ell$ such that the population size $N_\ell$ and the sample size $n_\ell$ tend to $+\infty$ as $\ell \rightarrow +\infty$. Thereafter, index $\ell$ is omitted in order to make the notation less cluttered but the asymptotic results and convergences are understood to be as $\ell \rightarrow +\infty$. The following assumptions are considered:

\begin{enumerate}[({A}1):]
    \item   $\pi_{ij} - \pi_i \pi_j = \O{n}{N^2}$ for each $i, j \in U$, $i \neq j$.
    \item   $d_i = \O{N}{n}$ for each $i \in U$.
    \item   The data is MAR.
    \item   The following model holds:
            \begin{align}
                \mbox{m:} \quad y_i = \betab ^\top \xb_i + \varepsilon_i,
            \end{align}
            with $\E_m (\varepsilon_i) = 0$, $\E_m (\varepsilon_i\varepsilon_j) = \sigma^2 < +\infty$ if $i = j$ and 0 otherwise and where $\E_m (\cdot)$ denotes the expectation with respect to the model m.
    \item   The imputation design is exactly balanced, i.e.
            \begin{align}
                \sum_{i \in S_r} \phi_{ij}^{\bknn} \xb_i = \sum_{i \in S_r} \psi_{ij}^{\bknn} \xb_i,
            \end{align}
            for each $j \in S_m$.
\end{enumerate}

%
\begin{enumerate}[({A}1):] \setcounter{enumi}{5}
    \item   The approximation of the conditional imputation variance is exact, i.e.
            \begin{align}
                \Var_{I}(\widehat{Y}_I) = \Var^{app}_{I}(\widehat{Y}_I) = \sum_{i \in S_r} \sum_{ \substack{j \in S_m \\ \psi_{ij}^{\bknn} \neq 0}}
                                            c_{ij} d_j^2 \left( y_i - \bb^\top \xb_i   \right)^2,
            \end{align}
            where $c_{ij}$ and $\bb$ are defined below Formula~\eqref{var:imp:approx}.
    \item   $\#\left\{\psi_{ij}^{\bknn} \psi_{i\ell}^{\bknn} > 0  \left| i \in S_r \right. \right\} = \O{k^2}{n_m}$ for each $j, \ell \in S_m$ such that $j \neq \ell$. \\
            This hypothesis states that the way the donors distribute from one nonrespondent to another are not very dependent. This constraint is incompatible with the fact that the same respondents are always used as donors.
\end{enumerate}

\begin{prop}
    Suppose assumptions (A1) to (A7) hold. Then
    \begin{align}
        \frac{ \widehat{Y}_I - Y }{ N }
    \end{align}
    converges in probability to 0.
\end{prop}

The proof is given in the Appendix.

\section{Simulation study}\label{section:simulation}

A brief simulation study is conducted to test the performance of the new imputation method and to test the accuracy of
the proposed estimator for imputation variance.

\subsection{The data}

The MU284 population from~\cite{sar:swe:wre:92} was considered here. This data set is available in the \texttt{R sampling} package \citep{til:mat:07}. The following variables were considered (the initial names of the variables are written in brackets):
\begin{itemize}\setlength{\itemsep}{-3pt}
  \item $\yb$:  revenues from 1985 municipal taxation, in millions of kronor (\texttt{RMT85}),
  \item $\xb^1$: 1985 population, in thousands (\texttt{P85}),
  \item $\xb^2$: 1975 population, in thousands (\texttt{P75}),
  \item $\xb^3$: number of Conservative seats in municipal council (\texttt{CS82}).
\end{itemize}
The correlations between $\yb$ and the variables $\xb^1$, $\xb^2$ and $\xb^3$ are respectively $0.96$, $0.97$ and $0.52$.
The population size is $N=284$. Two cases were considered, namely
\begin{enumerate}[{Case} 1:]
    \item   The three auxiliary variables ($\xb^1$, $\xb^2$, and $\xb^3$) were considered,
    \item   Only the auxiliary variable that is the less correlated to $\yb$ (namely $\xb^3$) was considered.
\end{enumerate}
The model $m$ defined in Section~\ref{section:linear:model} induces a $R^2$ which is approximately 0.94 and 0.27 in Case 1 and in Case 2 respectively.

\subsection{Simulation settings}

A census was considered, which means that $\pi_i = d_i = 1$ for each unit $i$ of the population $U = \left\{1,2,\ldots,N\right\}$. The sample therefore matches the population, i.e. $S=U$.
One hundred respondents sets were created by generating 100 response indicator vectors $R$. Each component $r_i, i \in U$ of $R$ was generated from a Bernoulli distribution with parameter
\begin{align*}
    \theta_i = \frac{ 1 }{ 1 + \exp \left( 1 - \beta x_{i\ell} \right)},
\end{align*}
where $\beta$ is a positive coefficient used to reach the mean response rate 70\% (MAR), $x_{i\ell}$ is the value of the variable $\xb^\ell$ for unit $i \in U $, and $\ell = 1,3$ in Case 1 and in Case 2 respectively.


For each respondents set, 100 imputations were conducted with each of the following methods:
\begin{itemize}\setlength{\itemsep}{-3pt}
    \item   NNI: nearest-neighbor,
    \item   PMM: predictive mean matching proposed by~\cite{lit:88},
    \item   SRS: random hot-deck, donors randomly selected with replacement in the respondents set,
    \item   SRSWOR: same as SRS except that donors are selected without replacement as proposed in~\cite{kal:kis:81,kal:kis:84},
    \item   $k$NNI: $k$-nearest neighbor,
    \item   b$k$NNI: proposed method, balanced $k$-nearest neighbor,
\end{itemize}
with $k = 20$. For each imputation, the total, the 10th percentile, the 90th percentile, and the variance of the imputed variable of interest were estimated. Note that, for b$k$NNI, the matrix of imputation probabilities $\psib^{\bknn}$  was replaced by the matrix of imputation probabilities $\psib^{\knn}$ defined in Expression~\eqref{equation:psi:knn} for the simulations in which Algorithm~\ref{algorithm:matrix:imp:prob} failed to find a solution (see Section~\ref{section:choice:k}). Moreover, for each simulation, the imputation variance of the total obtained with the proposed method was estimated using Expression~\eqref{var:imp:approx}.

\subsection{Measures of comparison}

In order to measure the bias of the imputed estimator $\widehat{\theta}_I$ for a parameter $\theta$, the Monte Carlo relative bias $RB$ was considered. It is defined as
\begin{align*}\displaystyle
    \mbox{RB} \left(\widehat{\theta}_I\right) = \frac{ \widehat{\theta}_I^* - \theta}{\theta},
\end{align*}
where
\begin{align*}
  \widehat{\theta}_I^*  = \frac{1}{M_R}\frac{1}{M_I}\sum_{r = 1}^{M_R} \sum_{i = 1}^{M_I} \widehat{\theta}^{r,i}_I,
\end{align*}
$M_R = 100$ is the number of respondents sets generated, $M_I = 100$ is the number of imputations conducted for each respondents set, and $\widehat{\theta}^{r,i}_I$ is the estimate obtained for the $i$-th imputation of the $r$-th respondents set generated. The quantity $\widehat{\theta}_I^*$ therefore represents the mean of the estimated value of the parameter $\theta$ over the $M_R M_I$ simulations. The variability of the imputed estimator $\widehat{\theta}_I$ was measured through the Monte Carlo relative root mean square error (RRMSE) defined as
\begin{align*}
    \mbox{RRMSE} \left(\widehat{\theta}_I\right) =
            \frac{ \sqrt{   \mbox{MSE} \left(\widehat{\theta}_I\right) }}{\theta},
\end{align*}
where
\begin{align*}
    \mbox{MSE} \left(\widehat{\theta}_I\right) =
              \displaystyle \frac{1}{M_R}\frac{1}{M_I}\sum_{r = 1}^{M_R} \sum_{i = 1}^{M_I} \left( \widehat{\theta}^{r,i}_I - \theta \right)^2.
\end{align*}
Finally, the Monte Carlo relative root imputation variance (RRIV), or relative imputation standard deviation, of the imputed estimator $\widehat{\theta}_I$ was computed in order to measure the amount of variance due to imputation. It is defined as
\begin{align*}
    \mbox{RRIV}\left(\widehat{\theta}_I\right) = \frac{ \displaystyle \sqrt{ \mbox{IV}\left(\widehat{\theta}_I\right) }} {\theta},
\end{align*}
where
\begin{align*}
    \mbox{IV}\left(\widehat{\theta}_I\right) = \frac{1}{M_R}\sum_{r = 1}^{M_R} \frac{1}{M_I - 1} \sum_{i = 1}^{M_I} \left( \widehat{\theta}^{r,i}_I -  \widehat{\theta}_I^r   \right)^2,
\end{align*}
and
\begin{align*}
     \widehat{\theta}_I^r = \frac{1}{M_I} \sum_{i = 1}^{M_I} \widehat{\theta}^{r,i}_I
\end{align*}
represents the mean estimated value of $\theta$ for the $r$-th respondents set.

In order to test the accuracy of the variance formula of~Expression~\eqref{var:imp:approx}, the average over the simulations of the approximated conditional imputation variance was computed, namely
\begin{equation}
    \frac{1}{M_R}\sum_{r = 1}^{M_R} \Var^{app}_{I}(\widehat{Y}_I)^{r},
\end{equation}
where $\Var^{app}_{I}(\widehat{Y}_I)^{r}$ is the imputation variance obtained with Expression~\eqref{var:imp:approx} for the $r$-th respondents set generated. That one was then compared to the Monte Carlo imputation variance of the total $\mbox{IV}\left(\widehat{Y}_I\right)$ defined above.

\subsection{Results of the simulations}

Table~\ref{table:1} and Table~\ref{table:2} show measures of comparison for the six imputation methods considered in Case 1 and Case 2 respectively. Table~\ref{table:variance} displays the average over the simulations of the estimated imputation variance of the total as well as the Monte Carlo imputation variance of the total.


The results confirm that the proposed method (b$k$NNI) performs particularly well when there is a strong linear relation between the variable of interest and the auxiliary variable, as in Case 1 ($R^2 \approx 0.94$). Indeed, results of Table~\ref{table:1} show that, in Case 1, b$k$NNI outperforms the other donor imputation methods considered. It provides the smallest RB and the smallest RRMSE for each parameter of interest considered.

Moreover, Table~\ref{table:1} and Table~\ref{table:2} show that the neighborhood principle and the balancing principle have an effect on the imputation variance of the total. This effect depends on the strength of the relation between the variable of interest and the auxiliary variables. Indeed, in Case 1 (strong linear relation) RRIV of the total is 0.094 for SRS, which reduces to 0.008 for $k$NNI (neighborhood principle) and to 0.002 for b$k$NNI (neighborhood principle and balancing principle) whereas, in Case 2 (weak linear relation) these figures are 0.093, 0.019, and 0.016.

%

The results also show that selecting the donors without replacement (SRSWOR) among respondents induces a smaller imputation variance than selecting them with replacement (SRS), which is in agreement with~\cite{kal:kis:81,kal:kis:84}.

Finally, the results confirm that the performance of the proposed method relies on the strength of the linear relation that governs the data. Indeed, in Case~2 (Table~\ref{table:2}) this linear relation is much weaker than in Case 1 and the proposed method shows diminished performance compared to that observed in Case~1. Note that, in Case~2, the proposed method nevertheless still performs better overall than the other methods considered.


\begin{table}[!htb]\caption{Monte Carlo relative bias (RB), Monte Carlo relative root mean square error (RRMSE), and Monte Carlo relative root imputation variance (RRIV) for the total estimation, the 10-th percentile estimation, the 90-th percentile estimation, and the variance estimation of the variable of interest $\yb$ in Case 1.       \label{table:1}}
\begin{center}\small
\begin{tabular} {llrrr}
\toprule
                    &                & \multicolumn{3}{c}{Monte Carlo estimates} \\
                                     \cmidrule(r){3-5}
Parameter of interest&   Method      & RB    & RRMSE & RRIV   \\
\midrule
Total               &   NNI          & 0.008 & 0.010 & 0.000\\
                    &   PMM          & 0.015 & 0.017 & 0.000\\
                    &   SRS          & 0.281 & 0.297 & 0.094\\
                    &   SRSWOR       & 0.278 & 0.288 & 0.070\\
                    &   $k$NNI       & 0.030 & 0.032 & 0.008\\
                    &   b$k$NNI      &-0.001 & 0.003 & 0.002\\
\midrule
10-th percentile    &   NNI          & 0.067 & 0.093 & 0.011\\
                    &   PMM          & 0.039 & 0.093 & 0.010\\
                    &   SRS          & 0.187 & 0.120 & 0.040\\
                    &   SRSWOR       & 0.186 & 0.198 & 0.031\\
                    &   $k$NNI       & 0.098 & 0.124 & 0.046\\
                    &   b$k$NNI      & 0.006 & 0.083 & 0.053\\
\midrule
90-th percentile    &   NNI          & 0.002 & 0.009 & 0.000\\
                    &   PMM          & 0.005 & 0.013 & 0.000\\
                    &   SRS          & 0.246 & 0.256 & 0.068\\
                    &   SRSWOR       & 0.248 & 0.255 & 0.054\\
                    &   $k$NNI       & 0.009 & 0.018 & 0.015\\
                    &   b$k$NNI      & 0.000 & 0.006 & 0.005\\
\midrule
Variance            &   NNI          &-0.001 & 0.002 & 0.000\\
                    &   PMM          &-0.002 & 0.002 & 0.000\\
                    &   SRS          & 0.389 & 0.533 & 0.361\\
                    &   SRSWOR       & 0.379 & 0.466 & 0.267\\
                    &   $k$NNI       &-0.004 & 0.004 & 0.002\\
                    &   b$k$NNI      & 0.000 & 0.001 & 0.000 \\
\bottomrule
\end{tabular}
\end{center}
\end{table}

\begin{table}[!htb]\caption{Monte Carlo relative bias (RB), Monte Carlo relative root mean square error (RRMSE), and Monte Carlo relative root imputation variance (RRIV) for the total estimation, the 10-th percentile estimation, the 90-th percentile estimation, and the variance estimation of the variable of interest $\yb$ in Case 2.       \label{table:2}}
\begin{center}\small
\begin{tabular} {llrrr}
\toprule
                    &                & \multicolumn{3}{c}{Monte Carlo estimates} \\
                                     \cmidrule(r){3-5}
Parameter of interest&   Method      & RB    & RRMSE & RRIV   \\
\midrule
Total               &   NNI          &-0.001 & 0.030 & 0.017\\
                    &   PMM          & 0.000 & 0.030 & 0.017\\
                    &   SRS          & 0.207 & 0.230 & 0.093\\
                    &   SRSWOR       & 0.207 & 0.222 & 0.069\\
                    &   $k$NNI       & 0.004 & 0.030 & 0.019\\
                    &   b$k$NNI      &-0.001 & 0.028 & 0.016\\
\midrule
10-th percentile    &   NNI          &-0.013 & 0.080 & 0.047\\
                    &   PMM          &-0.012 & 0.080 & 0.047\\
                    &   SRS          & 0.092 & 0.108 & 0.036\\
                    &   SRSWOR       & 0.092 & 0.105 & 0.027\\
                    &   $k$NNI       & 0.023 & 0.076 & 0.046\\
                    &   b$k$NNI      & 0.005 & 0.074 & 0.045\\
\midrule
90-th percentile    &   NNI          & 0.004 & 0.051 & 0.034\\
                    &   PMM          & 0.005 & 0.052 & 0.034\\
                    &   SRS          & 0.193 & 0.211 & 0.072\\
                    &   SRSWOR       & 0.193 & 0.207 & 0.056\\
                    &   $k$NNI       & 0.004 & 0.053 & 0.036\\
                    &   b$k$NNI      &-0.001 & 0.052 & 0.034\\
\midrule
Variance            &   NNI          &-0.001 & 0.094 & 0.054\\
                    &   PMM          &-0.001 & 0.094 & 0.052\\
                    &   SRS          & 0.372 & 0.525 & 0.356\\
                    &   SRSWOR       & 0.376 & 0.473 & 0.268\\
                    &   $k$NNI       &-0.003 & 0.088 & 0.061\\
                    &   b$k$NNI      &-0.008 & 0.076 & 0.044\\
\bottomrule
\end{tabular}
\end{center}
\end{table}


The results in Table~\ref{table:variance} confirm that Formula~\eqref{var:imp:approx} can underestimate the imputation variance. The magnitude of this underestimation goes along with the strength of the linear relation between the variable of interest and the auxiliary variables. Indeed, in Case~2 (weak linear relation), the average approximated conditional imputation variance represents more than the 90\% of the Monte Carlo imputation variance of the total. This quantity drops to approximately 60\% in Case~1 (strong linear relation).

\begin{table}[htb]
\begin{center}
\caption{Average over the simulations of the approximated conditional imputation variance of Expression~(\ref{var:imp:approx}),
Monte Carlo imputation variance of the total and ratio of these two quantities in two different cases.\label{table:variance}}
\begin{tabular} {lrr}
\toprule
                                    & \multicolumn{2}{c}{Case} \\ \cmidrule(r){2-3}
                                    & 1    & 2    \\
\midrule
Average approx. imputation variance
                                    &  8172.74  &  1244589.00\\
Monte Carlo imputation variance
                                    & 13146.21  &  1327932.00\\
Ratio                               &     0.62  &        0.94\\
\bottomrule
\end{tabular}
\end{center}
\end{table}

\section{Conclusion}\label{section:conclusion}

In this paper, a new method of random hot-deck imputation, called balanced $k$-nearest neighbor, has been proposed. This method has the interesting property of being a donor imputation. It therefore produces observed and feasible values. The novelty of this method is that the selection of donors is viewed as a sampling problem and uses calibration and balanced sampling. Also, selection of donors is achieved in a nonparametric manner as donors are selected in neighborhoods of recipients. As this method is random, it can be used for total estimation as well as for quantiles and variance estimation.

This method offers the nice advantage that it produces a total estimator with negligible imputation variance and a quasi-null bias in specified cases. Indeed, the method involves three underlying models or principles. They provide conditions for the imputed total estimator to be an unbiased estimator and for the imputation variance of that estimator to cancel. The method is resistant to model misspecification in terms of bias but a unique model results in a quasi-null imputation variance of the total.

A formula to approximate the conditional imputation variance of the total has been suggested. The procedure used is inspired by that applied to estimate the variance of the total for balanced sampling. The proposed approximation tends to underestimate the conditional imputation variance of the total.

Finally, a simulation study has been conducted to test the performance of the proposed method and that of the approximation formula of conditional imputation variance. It has been confirmed that the new method performs well when a strong linear relation governs the data and that this performance decreases as this linear relation weakens. Lastly, it was confirmed that the formula for imputation variance of the total tends to underestimate the conditional imputation variance of this one. Note that the estimation of the variance due to the rounding problem is still an unresolved problem. This variance can also be approximated by multiple imputations.

\section*{Acknowledgements}
This research was supported by the Swiss Federal Statistical Office and by the Swiss National Science Foundation (project number P1NEP2\_151904). 

\clearpage

\newpage
\appendix

\section*{Proofs of the properties}\label{appendix:proof:thm}

\begin{proof}[Proof of Property~\ref{thm:linear:model}]~\\
    \begin{enumerate}
        \item   \begin{align*}
                    \E_m \E_I \left( \widehat{Y}_I - \widehat{Y} \right)
                        &=  \E_m  \left( \sum_{i \in S_r} d_i y_i + \sum_{j \in S_m} d_j  \sum_{i \in S_r} \psi^{\bknn}_{ij} y_i - \sum_{i \in S} d_i y_i \right)\\
                        &=  \left( \sum_{i \in S_r} d_i \betab^\top\xb_i + \sum_{j \in S_m} d_j  \sum_{i \in S_r} \psi^{\bknn}_{ij} \betab^\top\xb_i  - \sum_{i \in S} d_i \betab^\top\xb_i  \right)\\
                        &= \betab^\top \left[ \E_I \left( \widehat{\Xb}_I \right) - \widehat{\Xb} \right] = 0,
                \end{align*}
                where the last equality comes from Equation~\eqref{equation:constraint:psi}. As it is supposed that the data is MAR, the expectation with respect to $m$, the one with respect to $p$, and the one with respect to $q$ can be reversed. It therefore produces
                \begin{align*}
                    \Bias \left( \widehat{Y}_I \right) &=  \E_m \E_p \E_q \E_I \left( \widehat{Y}_I - Y \right)
                                                            =   \E_m \E_p \E_q \E_I \left( \widehat{Y}_I - \widehat{Y} + \widehat{Y} - Y \right)\\
                                                            &=  \E_m \E_p \E_q \E_I \left( \widehat{Y}_I - \widehat{Y} \right)
                                                            =   \E_p \E_q \E_m \E_I \left( \widehat{Y}_I - \widehat{Y} \right) = 0.
                \end{align*}
        \item   If $y_i = \betab ^\top \xb_i $,
                \begin{align*}
                    \widehat{Y}_I - \E_I \left(\widehat{Y}_I\right) = \betab^\top\widehat{\Xb}_I - \E_I \left(\betab^\top\widehat{\Xb}_I\right)
                                                                    = \betab^\top \left[ \widehat{\Xb}_I - \E_I \left(\widehat{\Xb}_I\right) \right] \approx 0,
                \end{align*}
                where the last approximation comes from Remark~\ref{rem:2} (page~\pageref{rem:2}).
                Therefore
                \begin{align*}
                    \Var_I  \left( \widehat{Y}_I \right) = \E_I \left[ \widehat{Y}_I - \E_I \left( \widehat{Y}_I\right) \right]^2 \approx 0,
                \end{align*}
                and
                \begin{align*}
                    \Var_{imp} = \E_p \E_q \Var_I  \left( \widehat{Y}_I \right) \approx 0.
                \end{align*}
    \end{enumerate}
\end{proof}

\begin{proof}[Proof of Property~\ref{thm:response:model}]
    \begin{align*}
        \E_I \left( \widehat{Y}_I \right)   &= \sum_{i \in S_r} d_i y_i + \sum_{j \in S_m} d_j \sum_{i \in S_r} \psi_{ij}^{\bknn} y_i
                                            = \sum_{i \in S_r} d_i \left( 1 + \sum_{j \in S_m} \frac{d_j}{d_i} \psi_{ij}^{\bknn} \right) y_i\\
                                            &= \sum_{i \in S_r} d_i \frac{1}{\theta_i} y_i.
    \end{align*}
    If $\theta_i$ is the true response probability, this last expression represents the propensity score adjusted estimator. Therefore
    \begin{align*}
        \E_p \E_q \E_I \left( \widehat{Y}_I \right) = \E_p \E_q \left( \sum_{i \in S_r} d_i \frac{1}{\theta_i} y_i \right) = Y,
    \end{align*}
    and consequently
    \begin{align*}
        \Bias \left( \widehat{Y}_I \right) =  \E_p \E_q \E_I \left( \widehat{Y}_I - Y \right) = 0.
    \end{align*}
\end{proof}

\begin{proof}[Proof of Property~\ref{thm:neighborhood:principle}]~\\
    We have
    \begin{align*}
        \E_I \left( \widehat{Y}_I - \widehat{Y} \right) &= \sum_{j \in S_m} d_j \sum_{i \in S_r } \psi_{ij}^{\bknn} y_i - \sum_{j \in S_m} d_j y_j \\
                                                        &= \sum_{j \in S_m} d_j \sum_{i \in S_r } \psi_{ij}^{\bknn} y_i - \sum_{j \in S_m} d_j \sum_{i \in S_r } \psi_{ij}^{\bknn} y_j \\
                                                        &= \sum_{j \in S_m} d_j \sum_{i \in S_r } \psi_{ij}^{\bknn} (y_i - y_j) \\
                                                        &= 0,
    \end{align*}
    where the last equality comes from the fact that $\psi_{ij}^{\bknn}$ is nonzero only if $i \in knn(j)$ and from the hypothesis of the property. Therefore, it produces
    \begin{align*}
        \Bias \left( \widehat{Y}_I \right) &=
                \E_p \E_q \E_I \left( \widehat{Y}_I - Y \right)
                = \E_p \E_q \E_I \left( \widehat{Y}_I  - \widehat{Y} + \widehat{Y} - Y \right) \\
                &= \E_p \E_q \E_I \left( \widehat{Y}_I - \widehat{Y} \right) = 0.
    \end{align*}
\end{proof}

\section*{Proof of the Proposition}\label{appendix:proof:prop}

The proof of the Proposition requires the four following lemmas.

\begin{lem}\label{lem:1}
    Suppose that assumptions (A1) and (A2) hold. Then $\frac{ \widehat{Y} - Y}{N}$
    converges in probability to 0.
\end{lem}

\begin{proof}
    As $\widehat{Y}$ is a design unbiased estimator of $Y$, we have
    \begin{align}
        \E_{p} \left( \frac{ \widehat{Y} - Y}{N} \right) = 0.
    \end{align}
    Moreover, we have
    \begin{align}
        \Var_p \left( \frac{ \widehat{Y} - Y}{N} \right)
            &= \frac{1}{N^2} \sum_{i \in U} \sum_{j \in U} \frac{ \pi_{ij} - \pi_i \pi_j }{ \pi_i \pi_j} y_i y_j    \\
            &\leq \frac{1}{N^2}\sum_{i \in U} \sum_{\substack{j \in U \\j \neq i}} \frac{ \pi_{ij} - \pi_i \pi_j }{ \pi_i \pi_j} y_i y_j + \frac{1}{N^2}\sum_{i \in U} \frac{ 1 }{ \pi_i } y_i^2\\
            &= \O{1}{n},
    \end{align}
    where the last equality follows from assumptions (A1) and (A2). By Bienayme-Chebychev inequality, we conclude that $\frac{ \widehat{Y} - Y}{N}$ converges in probability to 0.
\end{proof}

\begin{lem}\label{lem:2}
   Suppose that assumptions (A2) and (A6) hold. Then
    \begin{align}
        \Var_I  \left( \frac{ \widehat{Y}_I - \widehat{Y}}{N} \right) = \O{1}{n}.
    \end{align}
\end{lem}

\begin{proof}
 Assumption (A6) implies
    \begin{align}
        \Var_I \left( \frac{ \widehat{Y}_I - \widehat{Y}}{N} \right)
                &=  \frac{1}{N^2} \Var^{app}_I \left( \widehat{Y}_I \right)\\
                &=  \frac{1}{N^2} \sum_{i \in S_r} \sum_{ \substack{j \in S_m \\ \psi_{ij}^{\bknn} \neq 0}}
                                            c_{ij} d_j^2 \left( y_i - \bb^\top \xb_i   \right)^2\\
                &\leq \frac{1}{N^2} \sum_{j \in S_m} \sum_{ \substack{i \in S_r \\ \psi_{ij}^{\bknn} \neq 0}}
                                            \psi_{ij}^{\bknn} \frac{ n_m k } { n_m k - q} d_j^2 \left(y_i - \bb^\top \xb_i   \right)^2\\
                &= \O{1}{n}
    \end{align}
    where the last inequality comes from assumption (A2), from $y_i - \bb^\top \xb_i = O(1)$, and from $\psi_{ij}^{\bknn} = \O{1}{k}$.
\end{proof}

\begin{lem}\label{lem:3}
    Suppose that assumptions (A2), (A4), and (A7) hold. Then
    \begin{align}
        \Var_m \E_I  \left( \frac{ \widehat{Y}_I - \widehat{Y}}{N} \right) = \O{1}{n}.
    \end{align}
\end{lem}

\begin{proof}
    Using $\E_I \left( \psi_{ij}^{\bknn} \right) = \phi_{ij}^{\bknn}$ and assumption (A4), we get
    \begin{align}
        \E_I  \left( \frac{ \widehat{Y}_I - \widehat{Y}}{N} \right)
            &=  \frac{1}{N} \left( \sum_{j \in S_m}d_j \sum_{i \in S_r} \psi_{ij}^{\bknn} y_i  - \sum_{j \in S_m} d_j y_j\right)\\
            &=  \frac{1}{N} \left[ \sum_{j \in S_m}d_j \sum_{i \in S_r} \psi_{ij}^{\bknn} (\betab ^\top \xb_i + \varepsilon_i)  - \sum_{j \in S_m} d_j (\betab ^\top \xb_j + \varepsilon_j)\right].
    \end{align}
    Therefore, from assumption (A4) again, we obtain
    \begin{align}
        \Var_m \E_I  & \left( \frac{ \widehat{Y}_I - \widehat{Y}}{N} \right)\\
            &=  \frac{1}{N^2} \left[ \sum_{i \in S_r} \left(\sum_{j \in S_m} d_j  \psi_{ij}^{\bknn}\right)^2 \Var_m( \varepsilon_i )
                + \sum_{j \in S_m} d_j^2 \Var_m(\varepsilon_j)\right]\\
            &= \frac{1}{N^2} \sigma^2 \left[ \sum_{i \in S_r} \left(\sum_{j \in S_m} d_j  \psi_{ij}^{\bknn}\right)^2
                + \sum_{j \in S_m} d_j^2 \right] \\
            &= \frac{1}{N^2} \sigma^2 \left[
                \sum_{i \in S_r} \sum_{j \in S_m} d_j^2  {\psi_{ij}^{\bknn}}^2 +
                \sum_{i \in S_r} \sum_{j \in S_m} \sum_{\substack{\ell \in S_m \\ \ell \neq j}} d_j d_\ell \psi_{ij}^{\bknn}\psi_{i\ell}^{\bknn}
                + \sum_{j \in S_m} d_j^2 \right] \\
            &= \frac{1}{N^2} \sigma^2 \left[
                \sum_{j \in S_m} d_j^2 \sum_{i \in S_r} {\psi_{ij}^{\bknn}}^2 +
                \sum_{j \in S_m} \sum_{\substack{\ell \in S_m \\ \ell \neq j}} d_j d_\ell  \sum_{i \in S_r} \psi_{ij}^{\bknn}\psi_{i\ell}^{\bknn}
                + \sum_{j \in S_m} d_j^2 \right] \\
            &= \O{1}{n}
    \end{align}
    where the last equality follows from assumption (A2), from assumption (A7), and from $\psi_{ij}^{\bknn} = \O{1}{k}$.
\end{proof}

\begin{lem}\label{lem:4}
    Suppose that assumptions (A2) to (A7) hold. Then
    $\frac{ \widehat{Y}_I - \widehat{Y}}{N}$ converges in probability to 0.
\end{lem}

\begin{proof}
    Assumption (A4) implies that (see proof of Property~\ref{thm:linear:model})
    \begin{align}\label{eq1:lem3}
        \E_m \E_I \left( \frac{ \widehat{Y}_I - \widehat{Y}}{N} \right) = 0.
    \end{align}
    Then, by assumption (A3), we get
    \begin{align}
        \E_{mpqI} \left( \frac{ \widehat{Y}_I - \widehat{Y}}{N} \right) = \E_p \E_q \E_m \E_I \left( \frac{ \widehat{Y}_I - \widehat{Y}}{N} \right) = 0.
    \end{align}
    Moreover, from Lemma~\ref{lem:2} and Lemma~\ref{lem:3}, we have
    \begin{align}\label{eq2:lem3}
        \Var_m \E_I \left( \frac{ \widehat{Y}_I - \widehat{Y}}{N} \right) + \E_m \Var_I \left( \frac{ \widehat{Y}_I - \widehat{Y}}{N} \right) = \O{1}{n}.
    \end{align}
     Assumption (A3), Equation~\eqref{eq1:lem3} and Equation~\eqref{eq2:lem3} together imply
    \begin{align}
       \Var_{mpqI} \left( \frac{ \widehat{Y}_I - \widehat{Y}}{N} \right) = \Var_{pqmI} \left( \frac{ \widehat{Y}_I - \widehat{Y}}{N} \right) = \O{1}{n}.
    \end{align}
    By Bienayme-Chebychev inequality, we conclude that $\frac{ \widehat{Y}_I - \widehat{Y}}{N}$ converges in probability to 0.
\end{proof}

\begin{proof}[Proof of the Proposition]~\\
    The conclusion follows directly from equality
    \begin{align}
        \frac{ \widehat{Y}_I - Y}{N} = \frac{ \widehat{Y}_I - \widehat{Y}}{N} + \frac{ \widehat{Y} - Y}{N},
    \end{align}
    Lemma~\ref{lem:1}, and Lemma~\ref{lem:4}.
\end{proof}

\end{document}